\documentclass[journal]{IEEEtran}
\usepackage{amsmath,amssymb,amsthm}
\usepackage{graphicx}
\usepackage{booktabs}
\usepackage[hidelinks]{hyperref}

\theoremstyle{plain}
\newtheorem{theorem}{Theorem}
\newtheorem{proposition}{Proposition}
\newtheorem{corollary}{Corollary}
\theoremstyle{definition}
\newtheorem{definition}{Definition}
\newtheorem{assumption}{Assumption}
\newtheorem{remark}{Remark}

\newcommand{\E}{\mathbb{E}}
\newcommand{\Prob}{\mathbb{P}}
\newcommand{\Sset}{\mathcal{S}}
\newcommand{\dalpha}{d_{\alpha}}

\title{Conformal Recovery-Deadline Certificates for Runtime Assurance of
Adapting Controllers}
\author{Alireza Shojaei\thanks{A. Shojaei is with the Myers-Lawson School of Construction, Virginia Tech, Blacksburg, VA 24061 USA (e-mail: shojaei@vt.edu).}}

\begin{document}
\maketitle

\begin{abstract}
Runtime assurance (RTA) protects a safety-critical system by switching from an advanced
controller to a verified safe controller when a monitored condition is violated. The
standard \emph{latching} rule, which trips on the first breach of the safe set and then
coasts, is correct for a \emph{diverging} controller but pathological for a \emph{capable
online-adapting} one. Such a controller is unsafe \emph{by design} during a bounded
recovery transient (it must excite the plant to identify the fault before it can correct
it), so a latching shield trips on that transient and suppresses a controller that would
have recovered. We introduce the \emph{conformal recovery-deadline certificate}, a
split-conformal, distribution-free, finite-sample upper bound on the adapting controller's
recovery time that licenses delaying fallback with a coverage guarantee, backstopped by a
verified monitor at a hard critical limit. The certified deadline discriminates capable
from incapable controllers, keeping the recoverer autonomous while catching the diverger,
and the construction cleanly separates \emph{autonomy} (governed by the statistical
coverage) from \emph{safety} (governed by the verified backstop), an instance of
\emph{reliability-asymmetric} design. We prove marginal coverage, a weighted extension
that restores coverage under a known fault-distribution shift, and group-conditional
(Mondrian) coverage; we demonstrate all three on two unrelated Simplex testbeds, a 6-DOF
spacecraft attitude controller and a torque-controlled inverted pendulum. Both show the same
suppression pathology and the same cure, which is what makes the certificate a
domain-general \emph{mechanism} rather than a single-system trick.
\end{abstract}


\section{Introduction}
\label{sec:intro}
Runtime assurance buys safety for an unverifiable controller by watching a monitored
condition and switching to a verified fallback the instant that condition is breached. The
arrangement has a tacit premise. It assumes that a controller in breach is a controller in
trouble, so that taking authority away can only help. For the controllers runtime assurance
was built around, fixed laws and policies that either hold the safe set or diverge out of
it, the premise is exactly right. It stops being right for a controller that \emph{adapts
online}, and that is the gap this paper addresses.

A controller that adapts online to an unmodeled fault, inferring a latent fault parameter
from the closed-loop response and correcting for it, can recover from failures a fixed
controller cannot. Adaptation, however, carries a cost the runtime-assurance premise was not
written for. To identify the fault, the adapting controller must first \emph{act under the
wrong model}. On a control-direction reversal it drives the wrong way, leaves the safe set,
reads the resulting motion, and only then corrects. The state change it induces is precisely
the measurement that exposes the fault sign, so the excursion is not a malfunction but the
controller's diagnostic probe. This excursion is a \emph{bounded recovery transient},
temporary and, for a capable controller, self-correcting, yet it is indistinguishable, at
the instant of the first safe-set breach, from the onset of genuine divergence. A monitor
that decides on the present state alone cannot tell the two apart, because at that instant
they look identical.

Classical runtime assurance (RTA) and the Simplex architecture~\cite{Sha2001,Hook2016}
resolve the ambiguity in the only way a present-state monitor can, namely conservatively.
Once a monitored condition (the safe-set boundary) is breached, they switch to the verified
safe controller and \emph{latch}, coasting on the fallback thereafter. For a diverging
controller this is exactly right, since a diverging controller only gets worse and authority
should leave it for good. For a capable adapting controller the same rule is a
\emph{suppression pathology}. The shield trips on the recovery transient and never returns
control, converting a controller that would have recovered into a permanent safe-hold.
Worse, the fallback's coast removes the very excitation the adapting controller needed to
identify the fault, so the shield does not merely fail to help, it actively dismantles the
recovery it was reacting to. The capable controller is penalized precisely for being
informative, and the more diagnostic its probe, the harder the latch punishes it. We
emphasize at the outset that this is not a flaw in any particular shield. It is a structural
consequence of pairing a present-state latch with a controller whose competence expresses
itself as a transient breach, and identifying it is the first contribution of this paper.

The resolution does not require a better monitor or a tamer controller. It requires deciding
\emph{on the right quantity}. The property that actually distinguishes a recovering
controller from a diverging one is not the present breach but the \emph{recovery time}, the
duration the controller spends outside the safe set before it returns and stays. A recovering
controller has a finite, well-behaved recovery time; a diverging one has an effectively
infinite recovery time. If a deadline could be placed on recovery time with a guarantee, the
shield could license the bounded transient up to that deadline and still catch genuine
divergence at it. We construct exactly such a deadline. We replace the latch, and the ad-hoc
``trip after a fixed margin'' heuristics that practitioners reach for instead, with a
split-conformal upper bound $\dalpha$ on the adapting controller's recovery time, calibrated
on held-out fault episodes, together with an engagement rule that keeps the controller until
$\dalpha$ or until a verified backstop fires at a hard critical limit, whichever comes first.
The bound is distribution-free and finite-sample under exchangeability~\cite{Vovk2005,Lei2018,Angelopoulos2023},
assuming nothing about the controller, the plant, or the shape of the recovery-time
distribution. Its coverage level $1-\alpha$ is a tunable knob that trades autonomy against
the cost of a late escalation, but never against safety, because the verified backstop is the
floor that no statistical choice can lower.

The construction has a structure worth naming, because it is what makes it safe to deploy.
Autonomy, the retention of the capable controller through its recovery, is governed by the
statistical certificate and is therefore as reliable as the exchangeability assumption
behind it. Safety, the avoidance of the critical set, is governed by the verified backstop
and is therefore as reliable as the formal proof behind it. The two are deliberately not the
same object and do not share a failure mode. A miscalibrated certificate costs autonomy and
nothing else; safety holds unconditionally. We call this \emph{reliability-asymmetric}
design, in which a verified element is allowed to bound an unverified one precisely because
their reliabilities are unequal, and the conformal recovery-deadline certificate is a clean
instance of it. The certificate is permitted to be wrong without endangering the vehicle,
which is exactly the license that lets a distribution-free statistical object sit inside a
safety-critical loop.

\subsection{Contributions}
\begin{enumerate}
\item We introduce the conformal recovery-deadline certificate and the recovery-aware
engagement rule (Sec.~\ref{sec:cert}), with a clean \emph{autonomy-vs-safety separation} in
which coverage governs autonomy while a verified backstop governs safety
(Prop.~\ref{prop:safety}). This is a concrete instance of reliability-asymmetric design, in
which the verified element bounds the unverified one, and it is what lets a distribution-free
statistical object license delayed fallback without ever touching the safety floor.
\item We prove the supporting theory, namely marginal coverage with right-censoring at the
horizon (Thm.~\ref{thm:marginal}), a weighted extension that restores coverage under a
\emph{known} fault-distribution shift (Thm.~\ref{thm:weighted}), and group-conditional
(Mondrian) coverage (Thm.~\ref{thm:mondrian}). The construction is honest by design. When the
controller's recovery rate falls below the requested $1-\alpha$, the deadline returns
$+\infty$ and the certificate \emph{refuses to certify} rather than fabricating a finite
deadline it cannot justify. The marginal-versus-conditional distinction under censoring is
stated explicitly and turned into a measured quantity, because precisely scoping a coverage
claim is what makes it trustworthy.
\item We demonstrate generality on two unrelated plants, exhibiting the same pathology and
the same cure on a 6-DOF spacecraft attitude controller and a torque-controlled inverted
pendulum (Sec.~\ref{sec:exp}), which elevates the certificate from a single-system fix to a
domain-general mechanism. On the pendulum we own the fault sampler, so the shift likelihood
ratio is \emph{exact}, and there we exhibit the weighted certificate restoring coverage that
the unweighted one provably loses, alongside a head-to-head against latching, heuristic, and
conformal-on-safety-value~\cite{Tabbara2025} baselines.
\end{enumerate}

\section{Related Work}
\label{sec:related}
The construction draws on two traditions that have matured largely apart, namely runtime
assurance for safety-critical control and distribution-free conformal inference. The
recurring theme of this section is that the ingredients exist on both sides but their
\emph{composition}, certifying a controller's recovery time and using that certificate as a
runtime engagement deadline, is what is new, and that the gap it fills is one the prior work
did not have to confront because it had not paired a verified shield with an online-adapting
controller.

\subsection{Runtime assurance and the Simplex architecture} The Simplex
architecture~\cite{Sha2001} established the template the field still uses. A complex,
high-performance controller is allowed to drive the plant only while a monitored safety
condition holds, and authority reverts to a simple verified safety controller the moment the
condition is violated. Run-time assurance (RTA) generalizes this into both an architectural
pattern and a certification practice~\cite{Hook2016,ASTMF3269}, and shielded reinforcement
learning~\cite{Alshiekh2018} imports the same switching logic into the training loop so that
an exploring agent can never commit an unsafe action. The design choice that concerns us is
what happens \emph{after} a switch. The standard rule is to \emph{latch}, holding the
verified controller once it engages, which is documented practice for spacecraft
operations~\cite{Dunlap2022}. Latching is the correct response to divergence and the wrong
response to a recoverable transient, and the Neural Simplex architecture~\cite{Phan2020}
already recognized half of this by adding \emph{reverse switching}, returning authority to
the advanced controller once the state is safe again. Reverse switching addresses
\emph{re-engagement after the fact}; our certificate addresses the prior question of
\emph{whether to disengage at all}, by licensing the bounded breach up front for a certified
duration rather than reacting to it and then recovering. The two are complementary, and on an
open-loop-unstable plant where a control-direction reversal disables the safe controller too,
reverse switching is unavailable for a reason we make precise in Sec.~\ref{sec:exp}.

\subsection{Safety filters and control barrier functions} A parallel tradition enforces
safety not by switching controllers but by minimally correcting the commanded action.
Control-barrier-function safety filters~\cite{Ames2019} project the nominal command onto the
set of inputs that keep a barrier certificate non-negative, rendering a designated safe set
forward-invariant. This is a powerful and now-standard guarantee, and it is also an
instructive contrast for the present setting. A barrier filter is built to keep the state
\emph{inside} the safe set at all times, which is precisely the behavior an online-adapting
controller cannot exhibit during fault identification, because its diagnostic probe requires
leaving the safe set transiently. A filter that forbids the excursion forbids the
identification, so it suppresses the capable controller for the same structural reason a
latch does. Our backstop therefore plays the barrier role only at the \emph{critical} limit,
where the safe controller must take over, and we deliberately leave the warning-to-critical
recovery zone open for the adapting controller to work in. The barrier guarantee and the
recovery-deadline guarantee certify different objects, set forward-invariance versus
return-time, and are best read as addressing different layers of the same problem.

\subsection{Conformal prediction} Conformal prediction supplies the coverage machinery. The
formal foundation is the framework of Vovk, Gammerman, and Shafer~\cite{Vovk2005}, which
turns exchangeability into finite-sample, distribution-free prediction sets, and the
split-conformal (inductive) variant we use, with its modern treatment for
regression~\cite{Lei2018} and accessible exposition~\cite{Angelopoulos2023}, trades a small
amount of statistical efficiency for a single held-out calibration pass and a deployment-time
cost of one quantile lookup. Two extensions are central here. Weighted conformal
prediction~\cite{Tibshirani2019} restores coverage under a known covariate shift by
reweighting the calibration distribution, which is exactly the tool needed when the
deployment fault mix differs from the calibration mix. Mondrian (group-conditional) conformal
prediction~\cite{Vovk2005} restores per-class coverage by calibrating within a priori
classes, which is what a marginal guarantee over a fault mixture cannot provide on its own.
We use these as machinery rather than claim them as novelty; the contribution is the runtime
\emph{object} they are applied to.

\subsection{Conformal runtime assurance and the recovery-time gap} The closest prior work
makes the switching decision itself statistical. Recent methods gate fallback with conformal
prediction on a learned \emph{safety value}~\cite{Tabbara2025} or adapt the shield's
parameters online through hidden-parameter inference with conformal bounds~\cite{Kwon2025}.
These certify a safety \emph{value} or a shield parameter in order to decide \emph{whether}
to fall back at the current instant. Our certified object is different, and the difference is
the point. We bound the controller's \emph{recovery time} and use that bound as the
engagement deadline, leaving the verified monitor unchanged. For a diverging controller the
useful question is whether it is currently in breach, which a value gate answers well. For an
online-adapting controller the useful question is whether it will \emph{return in time}, and
that is a property of the recovery-time distribution, not of the instantaneous safety value,
because the adapting controller is unsafe by design at the present instant and safe in
expectation over the transient. Casting recovery time as a horizon-censored last-exit time
(Sec.~\ref{sec:setup}), certifying it distribution-free,
and separating the resulting autonomy guarantee from the verified safety floor is, to our
knowledge, a combination the runtime-assurance literature has not previously assembled. The
two-domain demonstration that closes the paper is what shows it to be a mechanism rather than
a coincidence of one plant.

\section{Problem Setup}
\label{sec:setup}
We now make the setting precise. A system runs an \emph{advanced} controller $\pi$ (here, an
online-adapting fault-tolerant controller) under a Simplex-style shield with access to a
verified \emph{safe} controller
and a monitored \emph{safe set} $\Sset$ (e.g.\ pointing error within a warning threshold).
A \emph{fault episode} begins at $t=0$; the trajectory may leave $\Sset$ transiently while
$\pi$ identifies the fault. Let $e_t \ge 0$ be a scalar safety signal (distance outside
$\Sset$; $e_t=0$ inside), observed at discrete steps $t=0,\dots,H$ over a finite horizon
$H$. A larger \emph{critical} set boundary $\Phi_c$ marks states from which safety is lost;
a \emph{verified backstop} monitors $e_t \ge \Phi_c$ and is sound by construction (in our
testbeds, a Kind~2-checked monitor~\cite{Champion2016}).

The single quantity the certificate is built around is how long the trajectory spends in
breach before it returns and \emph{stays}. We define recovery time as a last-exit time rather
than a first-return time deliberately, because what matters operationally is not the first
moment the controller dips back into the safe set, which an oscillating transient may do
repeatedly, but the moment after which it never leaves again within the horizon. A
controller that touches the safe set and then exits once more has not recovered, and the
last-exit definition is what captures that correctly.

\begin{definition}[Recovery time]
\label{def:recovery}
The recovery time of an episode with safety trace $(e_0,\dots,e_H)$ is
\[
  \tau \;=\;
  \begin{cases}
    0, & \text{if } e_t \in \Sset \ \forall t \le H \ (\text{never breaches}),\\
    \ell+1, & \text{if } \ell := \max\{t \le H : e_t \notin \Sset\} < H,\\
    +\infty, & \text{if } e_H \notin \Sset \ (\text{still breaching at the horizon}).
  \end{cases}
\]
That is, $\tau$ is the last-exit step from $\Sset$ plus one, the first step after which the
trajectory \emph{remains} in $\Sset$ through the horizon, with a right-censoring atom at
$+\infty$ for episodes that have not recovered by $H$.
\end{definition}

We emphasize what $\tau$ is and is not. It is a measurable functional of the
\emph{complete} length-$H$ trajectory, evaluated once per finished episode; it is a last-exit
time, and deliberately \emph{not} a stopping time with respect to the forward filtration,
since deciding $\{\tau\le t\}$ would require knowing whether the trajectory re-breaches after
$t$. This is not a defect to be hedged. Conformal prediction needs only that each completed
episode yield a well-defined score in the totally ordered set $[0,H]\cup\{+\infty\}$, so no
adaptedness or stopping-time property is required, and none is claimed. The $+\infty$ atom is
essential, because a diverging
controller yields $\tau=+\infty$, and the certificate must treat it without assuming a
density. This is precisely where a naive numeric approach would fail. Recovery time is not a
real-valued quantity to which one can fit a Gaussian or even a finite-mean distribution,
because a non-trivial mass of episodes never recover, and any procedure that quietly drops or
imputes those episodes would manufacture a finite deadline out of data that does not support
one. By keeping $+\infty$ as a genuine largest element of a totally ordered set, the
certificate carries the non-recovering episodes through the calibration faithfully, and they
are exactly what force an honest $+\infty$ deadline when the recovery rate is too low to
certify. The ordering, not any arithmetic, is all the construction needs.

\begin{assumption}[Exchangeability]
\label{ass:exch}
The calibration episodes' recovery times $\tau_1,\dots,\tau_n$ and a future test episode's
$\tau_{n+1}$ are exchangeable (e.g.\ i.i.d.\ from a fixed fault distribution $P$), each
valued in the totally ordered set $[0,H]\cup\{+\infty\}$.
\end{assumption}

Exchangeability is the one assumption the marginal guarantee rests on, and it is worth
stating what it does and does not demand. It asks that the calibration faults and the next
deployed fault be drawn from the same distribution, with no requirement on the form of that
distribution, on the controller, or on the plant dynamics. It does \emph{not} ask for
independence in the strong sense, only the weaker symmetry that the joint law is invariant to
reordering. When the deployment fault mix departs from the calibration mix, exchangeability
breaks, and rather than abandon the guarantee we recover it by reweighting
(Thm.~\ref{thm:weighted}) when the shift is known, or restrict it to per-class statements
(Thm.~\ref{thm:mondrian}) when the relevant structure is the fault class. The assumption is
thus not a fragile precondition but the hinge on which the entire family of guarantees
turns, and the rest of Sec.~\ref{sec:cert} is organized around what each relaxation of it
buys.

\section{The Conformal Recovery-Deadline Certificate}
\label{sec:cert}
The certificate has three pieces, namely the deadline itself, the engagement rule that turns
the deadline into shield behavior, and the coverage guarantees that say what the deadline
delivers. We present them in that order, with the intuition for each made explicit, because
the construction is deliberately simple and its value lies in what the simplicity buys rather
than in any technical intricacy.

\subsection{Deadline} For miscoverage level $\alpha\in(0,1)$ and $n$ calibration recovery
times, define the split-conformal deadline
\begin{equation}
  \label{eq:deadline}
  \dalpha \;=\;
  \begin{cases}
    \tau_{(k)}, & k := \lceil (1-\alpha)(n+1)\rceil \le n,\\
    +\infty, & k > n,
  \end{cases}
\end{equation}
where $\tau_{(1)}\le\cdots\le\tau_{(n)}$ are the ordered calibration values (with $+\infty$
sorting last). When the controller's recovery rate is below $1-\alpha$, $k>n$ forces
$\dalpha=+\infty$, so the certificate \emph{refuses} rather than fabricates a finite
deadline. The same inequality fixes the smallest usable calibration set, since $k\le n$
requires $n\ge\lceil1/\alpha\rceil-1$ recoveries (for instance $n\ge19$ at $\alpha=0.05$),
below which no finite deadline can carry the requested rate and the certificate returns
$+\infty$ by the identical rule.

The deadline is nothing more than an order statistic of the held-out recovery times,
specifically the smallest deadline under which at least a $1-\alpha$ fraction of calibration
episodes recovered, with the conformal $+1$ correction that makes the guarantee
finite-sample rather than asymptotic. The construction therefore needs no model of the
controller and no fit to the recovery-time distribution, which is exactly its strength.
Anything it certifies, it certifies distribution-free. The $+\infty$ branch deserves emphasis
as a feature, not an edge case. If the controller recovers on fewer than a $1-\alpha$
fraction of calibration faults, then no finite deadline can honestly cover the next fault at
the requested rate, and the certificate reports this by returning $+\infty$. It declines to
issue a deadline it cannot back, which is precisely the behavior one wants from an object that
sits in a safety-critical loop. A certificate that always returned a finite number would be
hiding its own incapacity; ours surfaces it, and the engagement rule below treats an
$+\infty$ deadline correctly by deferring entirely to the verified backstop. The refusal is
the certificate telling the truth about a controller that is not good enough to be trusted
with delayed fallback, and that honesty is what makes the finite deadlines it does issue
worth trusting.

\begin{theorem}[Marginal coverage]
\label{thm:marginal}
Under Assumption~\ref{ass:exch}, the deadline \eqref{eq:deadline} satisfies
$\Prob(\tau_{n+1} \le \dalpha) \ge 1-\alpha$. If the $\tau_i$ are almost surely distinct,
$\Prob(\tau_{n+1}\le\dalpha) \le 1-\alpha + 1/(n+1)$.
\end{theorem}
\begin{proof}
By exchangeability the rank of $\tau_{n+1}$ among $\{\tau_1,\dots,\tau_{n+1}\}$ is uniform
on $\{1,\dots,n+1\}$ (ties broken uniformly at random, or by any exchangeable rule). The
event $\{\tau_{n+1}\le \tau_{(k)}\}$ contains the event that this rank is at most $k$, whose
probability is at least $k/(n+1) = \lceil(1-\alpha)(n+1)\rceil/(n+1) \ge 1-\alpha$. The
order on $[0,H]\cup\{+\infty\}$ is total, so the argument uses no continuity and the
$+\infty$ atom is handled as an ordinary largest element; when $k>n$, $\dalpha=+\infty$ and
the claim is trivial since $\tau_{n+1}\le+\infty$ always. The upper bound is the standard
no-ties quantile-inversion bound.
\end{proof}

\begin{remark}[The operative bound is the lower one]
\label{rem:discrete}
Recovery times here are integer-valued on $\{0,1,\dots,H\}$ together with the $+\infty$ atom,
so the $\tau_i$ are \emph{not} almost surely distinct and the no-ties upper bound of
Theorem~\ref{thm:marginal} does not apply, because ties are generic, so coverage is only
guaranteed to be at least $1-\alpha$ and may exceed it. The over-coverage is bounded by the
mass of the largest tied recovery-time value, which is small whenever recovery times are
spread across many horizon steps, and the experiments (Sec.~\ref{sec:exp}) confirm coverage
sits close to the nominal target rather than far above it. We state the lower bound as the
operative guarantee and read the empirical coverage as the check on how much, if any,
discreteness slack the construction actually carries.
\end{remark}

\begin{remark}[Scope of the marginal guarantee]
\label{rem:marginal}
Theorem~\ref{thm:marginal} is \emph{marginal} coverage over the joint law of calibration
and test, averaging over the draw of the calibration set and over fault classes. It is
\emph{not} coverage conditional on a fixed calibration set, nor conditional on the fault
class. A deadline calibrated on a fault mix can under-cover a hard sub-class even while
meeting its marginal target (Sec.~\ref{sec:exp} measures exactly this), and the cure is
Theorem~\ref{thm:mondrian}. The censoring atom further means a controller that refuses to
engage (never breaches) trivially ``covers,'' so coverage must be read alongside the
recovery rate, never alone. We state this scope precisely because a coverage guarantee is
only as useful as the claim it actually makes, and the experiments report exactly the
conditional quantity the marginal theorem does not promise, turning the distinction into a
measurement rather than a hedge.
\end{remark}

\begin{corollary}[Training-conditional coverage]
\label{cor:training}
Condition on the calibration recovery times and let
$C=\Prob(\tau_{n+1}\le\dalpha\mid \tau_1,\dots,\tau_n)$ be the coverage the deployed deadline
actually delivers on future faults. For tie-free scores $C$ follows the law
$\mathrm{Beta}(k,\,n+1-k)$ with $k=\lceil(1-\alpha)(n+1)\rceil$, whose mean is
$k/(n+1)\ge1-\alpha$ and whose standard deviation is $O(n^{-1/2})$~\cite{Vovk2012}; for the
integer-valued recovery times here $C$ is stochastically at least this Beta variable. Hence
for any tolerance $\delta\in(0,1)$, a deadline guaranteeing $C\ge1-\alpha$ with
calibration-probability at least $1-\delta$ is obtained by raising the level to the
corresponding Beta quantile, a correction of order $\sqrt{\log(1/\delta)/n}$ over the
marginal level.
\end{corollary}

The marginal theorem controls the \emph{mean} of the deployed coverage, while
Corollary~\ref{cor:training} controls its lower tail, and a flight operator should ask for
the latter because a deployed certificate is computed once from the single calibration set in
hand rather than averaged over hypothetical re-draws. The two together let one trade a
slightly later deadline for a stated $(1-\delta)$ confidence that the coverage delivered in
flight meets target; the deadlines we report are the marginal-level ones, with the
training-conditional inflation available in closed form from the Beta quantile whenever an
operator wants to fix that confidence explicitly.

\subsection{Engagement rule and the safety/autonomy split} The deadline becomes shield
behavior through a single rule. The shield keeps $\pi$ engaged
while $e_t<\Phi_c$ and $t<\dalpha$; it escalates to the safe controller at the first time
$e_t\ge\Phi_c$ (verified backstop) or at $t=\dalpha$ if the episode has not recovered. The
rule has exactly two exit conditions, and they carry different kinds of authority. The
critical-limit exit is a hard physical gate enforced by the verified backstop, while the
deadline exit is the statistical certificate's judgment that the controller has had its
certified time and not used it. A recovering controller returns to the safe set well before
either fires and is never disturbed; a diverging controller either reaches $\Phi_c$ and is
caught by the backstop or runs out its deadline and is caught there. The rule thus admits the
capable controller and rejects the incapable one using the very quantity, recovery time, that
distinguishes them. One subtlety is
worth stating plainly. The deadline is calibrated on the offline last-exit $\tau$, which is
known only at the horizon, whereas the runtime rule must decide from the current state, and
the two notions can disagree only when a trajectory re-breaches after appearing recovered.
That disagreement changes only whether $\pi$ is retained, never whether the vehicle is safe,
because the verified backstop governs the $\Phi_c$ boundary independently of either notion of
recovery. The autonomy guarantee is therefore stated with respect to permanent recovery,
and the online rule is at worst conservative against it.

\begin{proposition}[Safety is verified; autonomy is certified]
\label{prop:safety}
Suppose the backstop soundly detects $e_t\ge\Phi_c$ and the safe controller renders
$\{e\ge\Phi_c\}$ avoidable once engaged. Then, under the engagement rule, the following hold.
(i) \emph{Safety} holds for \emph{every} episode regardless of $\alpha$ or of whether
Assumption~\ref{ass:exch} holds. (ii) \emph{Autonomy} (the retention of $\pi$ through
recovery) holds on at least a $1-\alpha$ fraction of recovering test episodes in
expectation (Thm.~\ref{thm:marginal}).
\end{proposition}
\begin{proof}
(i) The escalation time is at most the backstop's detection time of $\{e\ge\Phi_c\}$, which
is sound by hypothesis and independent of the conformal construction; hence the critical set
is avoided on every episode. (ii) An episode is escalated early (before recovery) only if
$\tau_{n+1}>\dalpha$ and it has not yet reached $\Phi_c$; by Thm.~\ref{thm:marginal}
$\Prob(\tau_{n+1}\le\dalpha)\ge1-\alpha$.
\end{proof}

Proposition~\ref{prop:safety} is the reliability-asymmetric core, and it is what makes the
whole construction deployable. A miscalibrated deadline
(wrong $\alpha$, or a violated exchangeability assumption) costs \emph{autonomy}, in the form
of a needless or delayed escalation, but never \emph{safety}, because the verified backstop
is the floor. The statistical certificate is permitted to be wrong; the verified monitor is
not. This is the precise sense in which a distribution-free statistical object can be placed
inside a safety-critical loop without inheriting the burden of formal verification. Its
errors are routed entirely into the autonomy budget, where they are recoverable, and away
from the safety budget, where they would not be. The two guarantees are proved under
different assumptions, fail in different ways, and are never asked to cover for each other,
and that separation is the design discipline the rest of the paper instantiates.

\subsection{Known fault-distribution shift} Deployment fault mixes rarely equal calibration
mixes, and a guarantee that silently assumed they did would be brittle in exactly the
situation a flight operator cares about. We therefore handle the shift head-on. If the test
episode is drawn from $Q\ll P$ with known likelihood ratio
$w=\mathrm{d}Q/\mathrm{d}P$, weighted conformal prediction~\cite{Tibshirani2019} restores
coverage by reweighting each calibration point by how much more likely it is under
deployment than under calibration. Let $p_i = w(x_i)/\big(\sum_{j=1}^n w(x_j)+w_{\star}\big)$
be normalized weights on
the calibration values (with $w_\star$ the test-point weight) and define $\dalpha^{w}$ as the
$(1-\alpha)$-quantile of $\sum_i p_i\,\delta_{\tau_i} + p_\star\,\delta_{+\infty}$. The
weighted deadline is the same order-statistic idea applied to a reweighted empirical
distribution, so it inherits the distribution-free character of the unweighted version while
correcting for the mix change.

\begin{theorem}[Coverage under known shift]
\label{thm:weighted}
If $(x_i,\tau_i)$ are i.i.d.\ from $P$ and $(x_{n+1},\tau_{n+1})$ from $Q\ll P$ with exact
$w=\mathrm{d}Q/\mathrm{d}P$, then with the per-test-point weight $w_\star=w(x_{n+1})$,
$\Prob_{Q}(\tau_{n+1}\le \dalpha^{w}) \ge 1-\alpha$.
\end{theorem}
\begin{proof}
This is the weighted-exchangeability result of Tibshirani et al.~\cite{Tibshirani2019}. The
weighted empirical distribution above is the conformal predictive distribution under the
covariate shift, and the rank argument of Thm.~\ref{thm:marginal} carries through with the
$p_i$ in place of uniform mass.
\end{proof}

\begin{remark}[Batch deadline vs.\ per-point guarantee]
\label{rem:batch}
A runtime deadline must be fixed \emph{before} the next fault is seen, so its weight
$w_\star$ cannot be the (unknown) test point's. In practice we set $w_\star$ to the
deployment-averaged weight $\E_Q[w]$, which targets the $Q$-quantile of recovery time; the
finite-sample per-point guarantee of Thm.~\ref{thm:weighted} is the idealization. We
therefore \emph{validate} the batch deadline's coverage on $Q$ empirically
(Sec.~\ref{sec:exp}, Table~\ref{tab:pend}); the effective sample size
$\widehat{n}_{\mathrm{eff}}=(\sum_i w_i)^2/\sum_i w_i^2$ reports how few calibration points
the shift leaves effective. This is the correct engineering posture for a runtime deadline,
namely to compute it from the deployment-averaged weight and then confirm the resulting
coverage on the deployment distribution directly, with the effective sample size disclosed as
an explicit measure of how much the shift has cost the calibration.
\end{remark}

\begin{theorem}[Group-conditional (Mondrian) coverage]
\label{thm:mondrian}
Partition faults into a priori classes $g$. Applying \eqref{eq:deadline} within class $g$
(deadline $\dalpha^{(g)}$ from the $n_g$ calibration recovery times of that class) gives
$\Prob(\tau_{n+1}\le \dalpha^{(g)}\mid \text{class}=g)\ge 1-\alpha$ for each $g$ with
$n_g\ge \lceil(1-\alpha)/\alpha\rceil$; otherwise $\dalpha^{(g)}=+\infty$.
\end{theorem}
\begin{proof}
Within each class the recovery times are exchangeable, so Thm.~\ref{thm:marginal} applies
conditionally on the class. The sample-size condition is $\lceil(1-\alpha)(n_g+1)\rceil\le
n_g$.
\end{proof}

The three guarantees form a graded response to how much is known about the deployment
distribution. When the deployment mix matches calibration, the marginal deadline
(Thm.~\ref{thm:marginal}) suffices and is the cheapest to compute. When the mix differs and
the shift is known, the weighted deadline (Thm.~\ref{thm:weighted}) restores coverage exactly.
When the relevant structure is the fault class and per-class coverage is what matters, the
Mondrian deadline (Thm.~\ref{thm:mondrian}) delivers it at the cost of needing enough
calibration episodes per class, and refuses with $+\infty$ for any class too sparsely
observed to certify. The same order-statistic construction underlies all three, which is why
the autonomy guarantee tracks the deployment assumptions cleanly rather than degrading in
some uncontrolled way when they weaken. The experiments exercise all three on real testbeds.

\section{Experiments}
\label{sec:exp}
We instantiate the certificate on two unrelated Simplex testbeds, and the purpose of using
two is specific. A result on a single plant would show that the certificate works there; the
same pathology and the same cure on two plants that share no dynamics, no actuator model, and
no controller code is what shows the certificate to be a domain-general \emph{mechanism}. The
spacecraft testbed carries the full closed-loop weight, being the deployed system of the
companion stack, while the pendulum testbed, where we own the fault sampler outright, is what
lets us exhibit the weighted and Mondrian guarantees against an \emph{exact} likelihood
ratio rather than an estimated one. Read together they cover both what a flight system can
demonstrate and what a fully controlled environment can prove.
Both testbeds define the safety
signal as a pointing/angle error, the safe set as a warning threshold, and recovery as
return-and-stay within it (Def.~\ref{def:recovery}); both have a verified backstop at a
critical angle. The adapting controller in each is an analytic feedback law fed an
\emph{online} estimate of the actuator fault, and the incapable baseline against which we
measure the pathology is the fault-unaware fixed controller. Coverage is reported as the mean
over repeated random calibration/test splits, which is an unbiased estimator of the marginal
coverage of Thm.~\ref{thm:marginal} and is the quantity that theorem actually guarantees.

\begin{figure*}[!t]\centering
\includegraphics[width=0.49\textwidth]{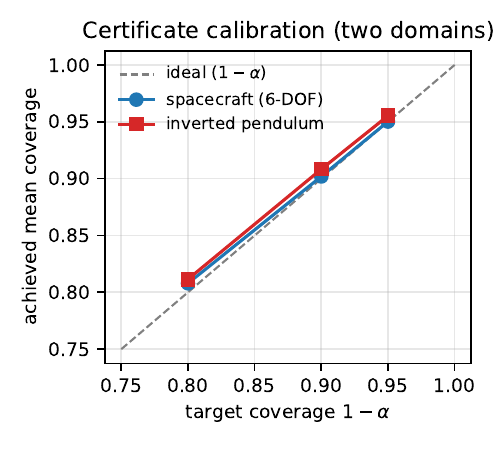}\hfill
\includegraphics[width=0.49\textwidth]{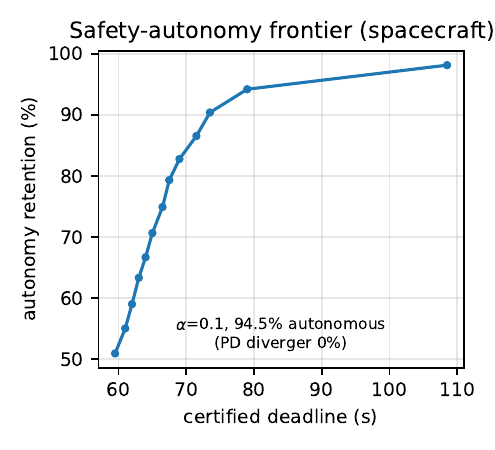}
\caption{\emph{Left.} Achieved vs.\ target coverage on both domains. Both the 6-DOF spacecraft
and the inverted pendulum track the $y{=}x$ ideal across $\alpha$, which is the evidence that
the certificate is calibrated on two plants that share no dynamics. \emph{Right.} The
spacecraft safety-autonomy frontier. A longer certified deadline (smaller $\alpha$) buys more
autonomy at the cost of trusting the controller longer in the breach zone, while the
fault-unaware diverger is caught at every operating point, so the frontier is a tunable
autonomy dial that never compromises the safety floor.}
\label{fig:calibration}
\end{figure*}

\subsection{Domain 1, 6-DOF spacecraft attitude control}
The first testbed is a Rapid-Motor-Adaptation controller on a 6-DOF Basilisk attitude
simulator recovering from unmodeled reaction-wheel faults. The safe set is $5^\circ$ pointing
error, and the verified backstop is at $\Phi_c=90^\circ$, a Kind~2-checked second-order
monitor that is sound by construction. This is the operationally serious case, because the
fault class is reversal-heavy gain faults on which the controller must transiently drive the
wrong way to identify the reversed wheel, which is exactly the regime that triggers the
suppression pathology under a latch. The controller is genuinely capable here, recovering to
$5^\circ$ on $99.8\%$ of held-out faults, with a median recovery of $119$ steps and a p95 of
$162$ steps, so the recovery-time distribution has the long but finite tail that a deadline
must accommodate.

\emph{The deadline is calibrated and tunable; the heuristic is neither.} The split-conformal
deadline meets its target coverage across all three operating points, achieving
$0.808/0.902/0.950$ at $\alpha=0.20/0.10/0.05$ over $300$ random calibration/test splits, so
the achieved coverage tracks the requested $1-\alpha$ to within sampling error
(Table~\ref{tab:spacecraft}). The contrast with the $p95\times1.3$ heuristic that earlier
reachability-RTA work reaches for is sharp. The heuristic over-covers at ${\sim}97\%$
regardless of the target, which means it is both untunable, since no choice of target moves
it, and unaccountable, since it carries no coverage guarantee at all. The conformal deadline
gives the operator a knob, and each setting of that knob comes with a distribution-free
promise about the resulting coverage, which is precisely what a hand-tuned margin cannot
offer.

\emph{The certificate discriminates capable from incapable.} Plugged into the engagement rule
at $\alpha=0.1$, where $\dalpha=152$ steps, the certificate keeps the recovering controller
$94.5\%$ autonomous, letting it work through the transient and retain authority, while the
fault-unaware controller, whose recovery time is effectively $+\infty$, misses the deadline
and is caught at it ($0\%$ autonomous). The same deadline thus admits the controller we want
and rejects the one we do not, using only the recovery time and the stated $1-\alpha$
guarantee. The choice of $\alpha$ traces a safety-autonomy frontier (Fig.~\ref{fig:calibration},
right), along which a smaller $\alpha$ buys a longer certified deadline and more autonomy at
the cost of trusting the controller longer in the breach zone, and the diverger is caught at
every operating point on that frontier. All numbers here are re-derived from the reproduced
Paper~A manifest (program/conformal\_rta.py).

\begin{table*}[!t]\centering\small
\caption{Domain 1 (spacecraft). Split-conformal deadline coverage vs.\ the heuristic, over 300 calibration/test splits of held-out gain faults.}
\label{tab:spacecraft}
\begin{tabular}{lccc}
\toprule
$\alpha$ & target $1-\alpha$ & conformal mean coverage & median deadline \\
\midrule
0.20 & 0.80 & 0.808 & 137 steps (68.5 s) \\
0.10 & 0.90 & 0.902 & 147 steps (73.5 s) \\
0.05 & 0.95 & 0.950 & 170 steps (85.0 s) \\
\midrule
\multicolumn{4}{l}{\footnotesize heuristic $p95\times1.3$ gives ${\sim}0.97$ coverage at every target (no guarantee, not tunable)} \\
\bottomrule
\end{tabular}
\end{table*}

\subsection{Domain 2, inverted pendulum (mechanism generality)}
The second testbed shares nothing with the first except the certificate. A torque-controlled
inverted pendulum (unstable upright) carries an actuator
effectiveness/sign fault $b$ on the control (applied torque $=b\,u$), which is the 1-DOF
analog of the spacecraft gain fault. The adapting controller infers $\mathrm{sign}(b)$ online
from the command-versus-acceleration correlation under \emph{noisy} rate sensing, and then
recovers, while the fault-unaware fixed law diverges on a reversal. The noise is what makes
this domain a fair stress test of the mechanism, because under noisy rate sensing a single
transient breach is genuinely ambiguous, so a naive latch is tempting and the temptation is
exactly the pathology we are characterizing. The safe set is $5^\circ$, the verified backstop
is $\Phi_c=45^\circ$, and faults are a mix of benign magnitude errors and direction
reversals. The implementation is pure-NumPy and deterministic
(conformal\_cert/domains/pendulum.py), so every number below regenerates bit-for-bit. The
significance of this domain is that we own the fault sampler outright, which means the
distribution-shift likelihood ratio is \emph{exact} rather than estimated, and that is what
turns the weighted-coverage guarantee from a claim into a clean demonstration.

\emph{Same pathology, same cure.} Every reversal forces a transient breach of the safe set,
so a latching Simplex shield suppresses the capable controller, holding it to $0\%$ autonomy,
exactly as on the spacecraft. Recovery-aware engagement at $\alpha=0.1$ keeps it $83.7\%$
autonomous, which is the reversal-class coverage at the marginal deadline and is reported as
such rather than as the marginal figure (cf.\ Rem.~\ref{rem:marginal}, which predicts that a
deadline calibrated on the fault mix can under-cover the hard reversal sub-class). The
diverging fixed law is caught at $0\%$. Marginal coverage meets target, reaching
$0.81/0.91/0.96$ at $\alpha=0.20/0.10/0.05$. That the pathology and the cure reappear
unchanged on a plant with no shared dynamics is the central evidence that the certificate is
a mechanism, not an artifact of the spacecraft model.

\emph{Known shift (Thm.~\ref{thm:weighted}).} Because we own the fault sampler, the shift
likelihood ratio is exact, and we can therefore test the weighted guarantee against ground
truth. Shifting the deployment reversal fraction from $30\%$
(calibration) to $70\%$ makes the deployment distribution strictly harder, and the unweighted
deadline duly under-covers at $86.0\%$, below its $0.90$ target, which is the exchangeability
assumption failing exactly as the theory says it must. The weighted deadline restores coverage
to $93.5\%$, back above target, with an effective sample size
$\widehat{n}_{\mathrm{eff}}=115$ that honestly reports how much the shift has thinned the
$200$-episode calibration set. This is the clean, verifiable instance of the exchangeability
cure that the spacecraft testbed, with its unknown true fault law, could only flag rather than
demonstrate, and it confirms that the weighted certificate does precisely what
Thm.~\ref{thm:weighted} promises when the likelihood ratio is known. The per-class (Mondrian,
Thm.~\ref{thm:mondrian}) deadlines are $0$ steps for the benign class, which never breaches,
and $37$ steps for the reversal class, which is the correct refusal-to-pool behavior, since a
single deadline averaged over both classes would systematically misserve each.

\emph{Engagement-rule baselines.} Table~\ref{tab:pend} places the certificate head-to-head
against the alternatives an RTA designer would actually consider, on the reversal regime, with
thresholds calibrated on a held-out split and metrics on a disjoint test split, all under the
same verified $\Phi_c$ backstop. The reading is clean. Latching destroys autonomy, retaining
just $0.5\%$, which is the pathology in a single number. The $p95\times1.3$ heuristic recovers
autonomy fully at $100\%$ but offers no guarantee and no tunability, so it cannot be defended
at review and cannot be dialed to a risk posture. The two conformal rules recover autonomy
\emph{with} a guarantee, and the distinction between them is what each guarantee is about.
Ours places the guarantee on recovery time, which is the property that matters for an adapting
controller, retaining $88.0\%$ autonomy, while value-gating~\cite{Tabbara2025} places it on
the peak excursion, a complementary object, and retains $93.5\%$. The two are not competitors
so much as guarantees on different quantities, and an operator who cares about return-time
wants the recovery-time bound. Crucially, all four rules catch $100\%$ of divergence before
the backstop, so none trades away the safety the backstop enforces. We omit Neural-Simplex
reverse switching~\cite{Phan2020} with cause rather than silently. On an open-loop-unstable
plant a control-direction reversal disables the safe controller too, since it shares the
reversed actuator, so reverse switching degenerates to latching here and would add nothing but
a redundant row.

\begin{table*}[!t]\centering\footnotesize
\caption{Domain 2 (pendulum). RTA engagement rules on the reversal regime, reporting autonomy retention on the capable adapting controller, divergence caught (before the verified backstop) on the fault-unaware controller, mean time-to-catch, and coverage guarantee.}
\label{tab:pend}
\setlength{\tabcolsep}{4pt}
\begin{tabular}{lcccl}
\toprule
Engagement rule & autonomy & divergence caught & time-to-catch & coverage guarantee \\
\midrule
Latching Simplex & 0.5\% & 100\% & 4.5 steps & none (the pathology) \\
Heuristic $p95\times1.3$ & 100.0\% & 100\% & 14.9 steps & none (over-covers) \\
CP-on-safety-value~\cite{Tabbara2025} & 93.5\% & 100\% & 11.5 steps & $1-\alpha$ on peak $|\theta|$ \\
\textbf{Conformal recovery-deadline (ours)} & 88.0\% & 100\% & 14.9 steps & $1-\alpha$ on recovery time \\
\bottomrule
\end{tabular}
\end{table*}

\begin{figure}[!t]\centering
\includegraphics[width=\linewidth]{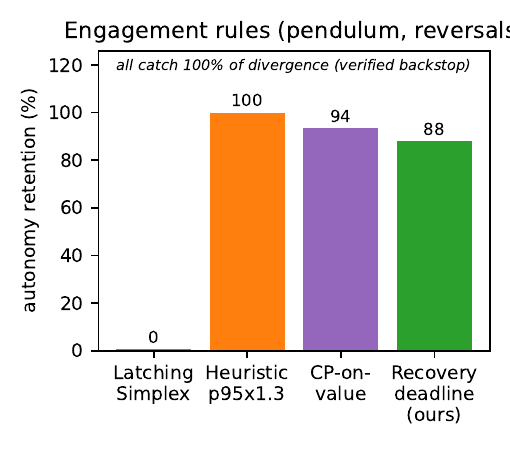}
\caption{Autonomy retention by engagement rule on the pendulum reversal regime. Latching
suppresses the capable controller (the pathology); the heuristic recovers autonomy but
without a guarantee; the two conformal rules recover it \emph{with} a guarantee. All four
catch $100\%$ of divergence at the verified backstop.}
\label{fig:engagement}
\end{figure}

\section{Discussion}
\label{sec:disc}
The result of this paper is a small one stated precisely, namely that the right object to
certify for an online-adapting controller under runtime assurance is its recovery time, and
that certifying it distribution-free reconciles a capability that a latch would suppress with
a safety floor a latch was protecting. We draw out three consequences, namely what
distribution-free recovery-deadline certification means for runtime assurance, what
generalizes across the two domains, and where the construction's authority comes from.

\subsection{Implications for runtime assurance} Runtime
assurance has always had to answer two questions, whether the protected controller is unsafe
now and what to do about it, and it has historically fused them into one present-state
switching decision. The pathology we identify is what happens when that fusion meets a
controller whose competence expresses itself as a transient breach. The decision to switch,
made on present state alone, fires on the very behavior that signals capability. Casting the
problem in terms of recovery time un-fuses the two questions. The verified backstop still
answers whether the state is unsafe \emph{now} at the critical limit, and the conformal
deadline separately answers whether the controller will \emph{return in time} from the
warning zone, and only the conjunction of those two answers triggers an escalation. This is
the methodological contribution beneath the mechanism. Distribution-free recovery-deadline
certification gives runtime assurance a principled way to grant a capable controller the
bounded breach it needs, without weakening the monitor and without hand-tuning a margin,
because the deadline carries a coverage guarantee the margin never could. The certificate's
willingness to return $+\infty$ when it cannot honestly certify is part of the same
discipline. It draws the line between controllers good enough to be trusted with delayed
fallback and controllers that are not, and it draws it from data rather than from optimism.

\subsection{Cross-domain generalization} The two testbeds share no dynamics, no
actuator model, and no controller code, and the certificate is the only component in common,
which is what licenses reading the agreement between them as evidence about the mechanism
rather than about either plant. Three things transfer. The \emph{pathology} transfers, since
latching suppresses the capable controller to near-zero autonomy on both the spacecraft and
the pendulum for the same structural reason, that a reversal forces a diagnostic breach. The
\emph{cure} transfers, since recovery-aware engagement restores autonomy on both while
catching every diverger at the backstop. And the \emph{coverage behavior} transfers, since
the achieved coverage tracks the target on both plants, $0.808/0.902/0.950$ on the spacecraft
and $0.81/0.91/0.96$ on the pendulum. Nothing in the construction depends on the dynamics
being attitude dynamics or pendulum dynamics, on the fault being a reaction-wheel reversal or
a torque-sign flip, or on the controller being a particular network, because the certificate
sees only a scalar safety signal and a recovery time derived from it. The implication is that
the recovery-deadline certificate is a drop-in for any Simplex-style system with an
online-adapting advanced controller and a verified backstop, and the two domains are the
existence proof of that portability.

\subsection{Source of the guarantee} The construction is sound enough to sit in a
safety-critical loop through the reliability-asymmetric separation of Prop.~\ref{prop:safety}.
Every quantity the certificate could get wrong, the coverage level, the exchangeability
assumption, the estimated likelihood ratio under shift, affects only autonomy, while safety is
delegated entirely to the verified backstop and is independent of all of it. This is why the
certificate may be distribution-free and assumption-light without that lightness ever reaching
the safety argument. The statistical layer is permitted to be approximate precisely because
the verified layer is not, and the two are composed so their failure modes never overlap. That
composition, rather than any single guarantee, is what we regard as the transferable lesson.

\subsection{Limitations}
\label{sec:limits}
We state the scope of the contribution as a set of boundaries on a solid result rather than as
a list of caveats. (i) \emph{The coverage is marginal.} The headline guarantee
(Thm.~\ref{thm:marginal}) is marginal over exchangeable faults, validated by mean coverage
over splits rather than per-fault or worst-case, and a deadline calibrated on a fault mix can
under-cover a hard sub-class (Rem.~\ref{rem:marginal}). We treat this as a property to measure
rather than to hide, reporting the reversal-class coverage explicitly and supplying the
Mondrian deadline (Thm.~\ref{thm:mondrian}) as the per-class remedy. (ii) \emph{Shift coverage
needs the likelihood ratio.} Coverage under deployment shift holds with a known or
well-estimated likelihood ratio (Thm.~\ref{thm:weighted}), and because a runtime deadline must
be fixed before the next fault is seen, we validate the batch deadline's shifted coverage
empirically rather than claim the idealized per-point bound (Rem.~\ref{rem:batch}). The
pendulum's exact-ratio shift is a controlled illustration in which the ratio is known by
construction, whereas flight deployment requires estimating it from telemetry, with the
effective sample size reported as the honesty check on how much the shift has cost. (iii)
\emph{The certificate is per controller and per fault class.} A recovery-deadline certificate
characterizes one \emph{deployed} controller's recovery distribution on a given fault class,
and a different controller or a materially different fault class has its own certificate that
must be calibrated on its own held-out episodes. None of these boundaries touches the safety
argument, which rests entirely on the verified backstop (Prop.~\ref{prop:safety}) and holds
for every episode regardless of $\alpha$, of the shift, or of whether exchangeability holds at
all. The limitations bound what the \emph{autonomy} guarantee promises, never the safety one,
which is the separation the design is built to preserve.

\section{Conclusion}
\label{sec:conclusion}
The latching rule that runtime assurance inherited from the diverging-controller case is a
suppression pathology for a capable online-adapting controller, which is unsafe by design
during a bounded recovery transient. Identifying that pathology is the first half of the
contribution, and it is a discovery about runtime assurance rather than a defect of any
shield, since it follows structurally from deciding to switch on present state alone when the
controller's competence shows up as a transient breach. The conformal recovery-deadline
certificate is the second half. It replaces
the latch, and the hand-tuned margins practitioners substitute for it, with a
distribution-free, finite-sample upper bound on recovery time that licenses
delaying fallback with a coverage guarantee, while a verified backstop holds the safety floor
unconditionally, and it reports an honest $+\infty$ when no finite deadline can be justified.
Autonomy is governed by statistics, safety by verification, and the two are composed so their
failure modes never meet. The same
pathology and the same cure on two unrelated plants is what makes the construction a
domain-general mechanism rather than a single-system trick, and it is what gives us
confidence that the recovery-deadline certificate is a reusable component for any
Simplex-style system pairing an online-adapting controller with a verified backstop.

\subsection{Reproducibility} Every number in the prose, tables, and figures derives from two
committed evidence files, \texttt{evidence/\allowbreak{}program/\allowbreak{}conformal\_rta.json} (Domain~1, in
Paper~A's reproduced manifest) and \texttt{evidence/\allowbreak{}conformal\_cert/\allowbreak{}pendulum\_domain.json}
(Domain~2, deterministic, regenerable via \texttt{python -m conformal\_cert.domains.pendulum}).
Tables~\ref{tab:spacecraft}--\ref{tab:pend} and
Figures~\ref{fig:calibration}--\ref{fig:engagement} are emitted from those files by
\texttt{tables.py} and \texttt{figures.py}, so the manuscript's quantitative claims are
manifest-reproducible.


\begin{thebibliography}{99}\small
\bibitem{Sha2001} L.~Sha. ``Using simplicity to control complexity.'' \emph{IEEE Software}, 2001.
\bibitem{Hook2016} J.~G.~Fuller, L.~Hook, N.~Hutchins, K.~N.~Maleki, M.~A.~Skoog. ``Toward run-time assurance in general aviation and unmanned aircraft vehicle autopilots.'' \emph{IEEE/AIAA Digital Avionics Systems Conference (DASC)}, 2016.
\bibitem{ASTMF3269} ASTM F3269-21. ``Standard practice for methods to safely bound behavior of aircraft systems containing complex functions using run-time assurance.''
\bibitem{Alshiekh2018} M.~Alshiekh, R.~Bloem, R.~Ehlers, B.~K\"onighofer, S.~Niekum, U.~Topcu. ``Safe reinforcement learning via shielding.'' \emph{AAAI}, 2018.
\bibitem{Ames2019} A.~D.~Ames, S.~Coogan, M.~Egerstedt, G.~Notomista, K.~Sreenath, P.~Tabuada. ``Control barrier functions: theory and applications.'' \emph{European Control Conference (ECC)}, pp.\ 3420--3431, 2019.
\bibitem{Phan2020} D.~T.~Phan, R.~Grosu, N.~Jansen, N.~Paoletti, S.~A.~Smolka, S.~D.~Stoller. ``Neural Simplex architecture.'' \emph{NASA Formal Methods (NFM)}, 2020.
\bibitem{Dunlap2022} K.~Dunlap, M.~Hibbard, M.~L.~Mote, K.~L.~Hobbs. ``Comparing run time assurance approaches for safe spacecraft docking.'' \emph{IEEE Control Systems Letters}, 6:1849--1854, 2022.
\bibitem{Tabbara2025} I.~Tabbara, Y.~Yang, H.~Sibai. ``Statistically assuring safety of control systems using ensembles of safety filters and conformal prediction.'' arXiv:2511.07899, 2025.
\bibitem{Kwon2025} M.~Kwon, T.~Ingebrand, U.~Topcu, L.~Feng. ``Adaptive shielding for safe reinforcement learning under hidden-parameter dynamics shifts.'' arXiv:2506.11033, 2025.
\bibitem{Vovk2005} V.~Vovk, A.~Gammerman, G.~Shafer. \emph{Algorithmic Learning in a Random World}. Springer, 2005.
\bibitem{Vovk2012} V.~Vovk. ``Conditional validity of inductive conformal predictors.'' \emph{Asian Conference on Machine Learning (ACML)}, PMLR 25:475--490, 2012.
\bibitem{Lei2018} J.~Lei, M.~G'Sell, A.~Rinaldo, R.~J.~Tibshirani, L.~Wasserman. ``Distribution-free predictive inference for regression.'' \emph{Journal of the American Statistical Association}, 113(523):1094--1111, 2018.
\bibitem{Angelopoulos2023} A.~N.~Angelopoulos, S.~Bates. ``Conformal prediction: a gentle introduction.'' \emph{Found.\ Trends Mach.\ Learn.}, 2023.
\bibitem{Tibshirani2019} R.~J.~Tibshirani, R.~Foygel Barber, E.~Cand\`es, A.~Ramdas. ``Conformal prediction under covariate shift.'' \emph{NeurIPS}, 2019.
\bibitem{Champion2016} A.~Champion et al. ``The Kind~2 model checker.'' \emph{CAV}, 2016.
\end{thebibliography}
\end{document}